\documentclass[12pt]{iopart}
\usepackage[dvips]{graphicx}
\usepackage{latexsym}
\usepackage{iopams}
\usepackage{amssymb}
\usepackage{verbatim,times,bbm}

\newtheorem{theorem}{Theorem}

\newtheorem{conjecture}[theorem]{Conjecture}

\newtheorem{proposition}[theorem]{Proposition}

\newenvironment{proof}[1][Proof]{\noindent\textbf{#1.} }{\ \rule{0.5em}{0.5em}}

\begin{document}

\title{Stochastic resonance in Gaussian quantum channels}

\author{Cosmo Lupo$^1$, Stefano Mancini$^{1,2}$, Mark M. Wilde$^3$}
\address{$^1$ School of Science and Technology, University of Camerino, I-62032 Camerino, Italy}
\address{$^2$ INFN Sezione di Perugia, I-06123 Perugia, Italy}
\address{$^3$ School of Computer Science, McGill University, H3A 2A7 Montreal, Qu\'ebec, Canada}
\ead{\mailto{cosmo.lupo@unicam.it}, \mailto{stefano.mancini@unicam.it},\mailto{mwilde@gmail.com}}

\begin{abstract}
We determine conditions for the presence of stochastic resonance in a lossy bosonic channel with 
a nonlinear, threshold decoding. The stochastic resonance effect occurs if and only if the 
detection threshold is outside of a \textquotedblleft forbidden interval.\textquotedblright\ 
We show that it takes place in different settings: when transmitting classical
messages through a lossy bosonic channel, when transmitting over an entanglement-assisted 
lossy bosonic channel, and when discriminating channels with different loss parameters. 
Moreover, we consider a setting in which stochastic resonance occurs in the transmission of a 
qubit over a lossy bosonic channel with a particular encoding and decoding. 
In all cases, we assume the addition of Gaussian noise to the signal and show that it does not
matter who, between sender and receiver, introduces such a noise. 
Remarkably, different results are obtained when considering a setting for private communication.
In this case the symmetry between sender and receiver is broken and the 
\textquotedblleft forbidden interval\textquotedblright\ may vanish, leading to the occurrence
of stochastic resonance effects for any value of the detection threshold.
\end{abstract}

\pacs{02.50.-r, 03.67.-a}

\section{Introduction}

Stochastic Resonance (SR) is a resonant phenomenon triggered by noise which can be described as a noise-enhanced 
signal transmission that occurs in certain non-linear systems~\cite{GHJM98}. It reveals a context where noise ceases 
to be a nuisance and is turned into a benefit. Loosely speaking one says that a (non-linear) system exhibits SR 
whenever noise {\it benefits} the system \cite{K06,PK09}. Qualitatively, the signature of an SR benefit is an 
inverted-U curve behaviour of the physical variable of interest as a function of the noise strength. It can take place 
in systems where the noise helps detecting faint signals. For example, consider a \emph{threshold detection} 
of a binary-encoded analog signal such that the threshold is set higher than the two signal values. If there is no noise, 
then the detector does not recover any information about the encoded signals since they are sub-threshold, and the same 
occurs if there is too much noise because it will wash out the signal. Thus, there is an optimal amount of noise that will 
result in maximum performance according to some measure such as signal-to-noise ratio, mutual information, or probability 
of success.

Recently the idea that noise can sometimes play a constructive role like in SR has started to penetrate the 
quantum information field too. In quantum communication, this possibility has been put forward in Refs.~\cite{Ting,BM,RGK05} 
and more recently in Refs.~\cite{W09,WK09,QKD,CHP10}. In this setting it has been shown that information theoretic quantities 
may ``resonate'' at maximum value for a nonzero level of noise added intentionally. It is then important to determine general 
criteria for the occurrence of such phenomenon in quantum information protocols.

Continuous variable quantum systems are usually confined to Gaussian states and processes, and SR effects are not expected 
in any linear processing of such systems.  
However, information is often available in digital (discrete) form, and therefore it must be subject to ``de-quantization'' 
at the input of a continuous Gaussian channel and ``quantization'' at the output~\cite{SJ67-GN98}. These processes are 
usually involved in the conversion of digital to analog signals and vice versa. Since these mappings are few-to-many and 
many-to-few, they are inherently non-linear, and similar to the threshold detection described above. We can thus expect the 
occurrence of the SR effect in this case. The simplest model representing such a situation is one in which a binary variable 
is encoded into the positive and negative parts of a real continuous alphabet and subsequently decoded by a threshold detection~\cite{KM03}. 
In some cases, one may not always have the freedom in choosing the threshold, and in such cases it becomes relevant to 
know that SR can take place. This may happen in {\it homodyne detection} if the square of the average signal times the 
overall detection efficiency (which accounts for the detector's efficiency, the fraction of the field being measured, etc.) 
is below the vacuum noise strength \cite{WM93}. It is also the case in discrimination between lossy channels, where the 
unknown transmissivities together with a faint signal make it unlikely to optimally choose the threshold value. 

In this paper, we consider a bit encoded into squeezed-coherent states with different amplitudes that are subsequently sent 
through a Gaussian quantum channel (specifically, a lossy bosonic channel \cite{EW07}). At the output, the states are subjected 
to threshold measurement of their amplitude. In addition to such a setting, we consider one involving entanglement shared by 
a sender and receiver as well as one involving quantum channel discrimination. Finally, we also consider the SR effect in 
quantum communication as well as in private communication. For all of these settings, we determine conditions for the occurrence 
of the SR effect. These appear as {\it forbidden intervals} for the threshold detection values. A ``forbidden interval'' (or region) 
is a range of threshold values for which the SR effect does not occur. We can illustrate this point by appealing again to the 
example of threshold detection of a binary-encoded analog signal. Suppose that the signal values are $A$ or $-A$ where $A>0$. 
Then if the threshold value $\theta$ is smaller in magnitude than the signal values, so that $ |\theta| \leq |A|$, the SR effect 
does not occur---adding noise to the signal will only decrease the performance of the system. In the other case where $ |\theta| > |A|$, 
adding noise can only increase performance because the signals are indistinguishable when no noise is present. As we said before, 
adding too much noise will wash out the signals, so that there must be some optimal noise level in this latter case. 
Our results extend those of Refs.~\cite{KM03,WK09} to other schemes. Remarkably, in the private communication scheme, the width 
of the forbidden interval can vanish depending on whether the sender or the receiver adds the noise. 
This means that in the former case the noise is always beneficial.

\section{Stochastic Resonance in Classical Communication}\label{sec:unassisted-comm}

Let us consider a lossy bosonic quantum channel with transmissivity 
$\eta\in(0,1)$ \cite{HW01,EW07}. 
Our aim is to evaluate the probability of successful decoding, considered 
as a performance measure, when sending classical information through such 
a channel. 

We consider an encoding of the following kind.
Let us suppose that the sender uses as input a displaced and squeezed vacuum.
Working in the Heisenberg picture, the input variable of the communication setup 
is expressed by the operator:
\begin{equation}\label{eq:qmessage1}
\hat{q}e^{-r}-\alpha_{q}(-1)^{X},
\end{equation}
encoding a bit value $X\in\{0,1\}$, where $\hat{q}$ is the
position quadrature operator, $\alpha_{q}\in\mathbbm{R}_{+}$ is the
displacement amplitude, and $r \geqslant 0$ is the squeezing parameter \cite{G05}.

Under the action of a    
lossy bosonic channel \cite{HW01} with transmissivity $\eta$ the input variable transforms
as follows:
\begin{equation}\label{lossych}
\sqrt{\eta} \left( \hat{q}e^{-r} - \alpha_{q}(-1)^{X} \right) + \sqrt{1-\eta} \, \hat{q}_{E} \, ,
\end{equation}
where $\hat{q}_{E}$ is the position quadrature operator of an environment mode
(assumed to be in the vacuum state for the sake of simplicity).

At the receiver's end, let us consider the possibility of adding a random,
Gaussian-distributed displacement $\nu_{q}\in\mathbbm{R}$,
with zero mean and variance $\sigma^2/2$, to the arriving state. 
Then, the output observable becomes as follows:
\begin{equation}\label{noiseB}
\sqrt{\eta} \left( \hat{q}e^{-r} - \alpha_{q}(-1)^{X} \right)  + \sqrt{1-\eta} \, \hat{q}_{E} + \nu_{q} \, .
\end{equation}
Notice that we could just as well consider the addition of noise at the sender's end. In that case, 
the last term $\nu_q$ of (\ref{noiseB}) would appear with a factor $\sqrt{\eta}$ in front.

Upon measurement of the position quadrature operator, the following signal
value $S_{X}$ results
\begin{equation}\label{signal1}
S_{X}=\sqrt{\eta} \left( qe^{-r} - \alpha_{q}(-1)^{X} \right) + \sqrt{1-\eta} \, q_{E} + \nu_{q} \, .
\end{equation}
Following Ref.~\cite{WK09}, we define a random variable summing up all noise
terms:
\begin{equation}\label{eq:all-noise1}
N \equiv \sqrt{\eta} \, q e^{-r} + \sqrt{1-\eta}\, q_{E} + \nu_{q} \, .
\end{equation}
Its probability density $P_N$ is the convolution of the probability densities of the 
random variables $qe^{-r}$, $\nu_{q}$ and $q_{E}$, these being independent of each other. 
Moreover, they are distributed according to Gaussian (normal) distribution, and so $P_N$ reads as
\begin{equation}\label{eq:noise-density1}
P_{N} = {\mathcal{N}}(0,\eta e^{-2r}/2) \circ {\mathcal{N}}(0,\left(  1-\eta\right)/2) \circ {\mathcal{N}}(0,\sigma^{2}/2) \, ,
\end{equation}
where $\circ$ denotes convolution, and
\[
{\mathcal{N}}\left(  \mu,K^{2}\right)  =\frac{1}{\sqrt{2\pi K^{2}}}
\exp\left[  \frac{-\left(  x-\mu\right)^2}{2K^{2}} \right] 
\]
denotes the normal distribution (as function of $x$) with mean $\mu$ and variance $K^{2}$.

Notice that the noise term (\ref{eq:all-noise1}) does not depend on the
encoded value $X$ and neither does its probability density. From
(\ref{eq:noise-density1}) we explicitly get
\begin{equation}\label{pX1}
P_{N} = {\mathcal{N}}(0,(1 - \eta + \eta e^{-2r} +\sigma^{2} )/2) \, .
\end{equation}

The output signal (\ref{signal1}) can now be written as
\[
S_{X} = N - \sqrt{\eta} \, \alpha_{q}(-1)^{X} \, .
\]
The receiver then thresholds the measurement result with a threshold $\theta \in \mathbbm{R}$
to retrieve a random bit $Y$ where
\begin{equation}\label{Y1}
Y\equiv H\left( N - \sqrt{\eta} \, \alpha_{q}(-1)^{X} - \theta \right)  \, ,
\end{equation}
and $H$ is the Heaviside step function defined as $H\left(  x\right)  =1$ if
$x \geqslant 0$ and $H\left(  x\right)  =0$ if $x<0$.

To evaluate the probability of successful decoding,  
we compute the conditional probabilities
\begin{eqnarray}
P_{Y|X}(0|0) &  = & \int_{-\infty}^{+\infty}\left[ 1 - H\left( n - \sqrt{\eta} \, \alpha_{q} - \theta \right) \right]  P_{N}\left(  n\right)  \;dn\nonumber\\
& = & 1-P_{Y|X}(1|0) \, , \\
P_{Y|X}(1|1) & = & \int_{-\infty}^{+\infty}H\left( n + \sqrt{\eta} \, \alpha_{q} - \theta\right) P_{N}\left(  n\right)  \;dn\nonumber\\
& = & 1-P_{Y|X}(0|1) \, .
\end{eqnarray}

Using (\ref{pX1}), we find
\begin{eqnarray}
P_{Y|X}(0|0) & = & \frac{1}{2} + \frac{1}{2} \mathrm{erf}\left[  \frac{ \theta + \sqrt{\eta} \, \alpha_{q} }{\sqrt{ 1 - \eta + \eta e^{-2r} + \sigma^{2} }}\right]  , \label{P00}\\
P_{Y|X}(1|1) & = & \frac{1}{2} - \frac{1}{2} \mathrm{erf}\left[  \frac{ \theta - \sqrt{\eta} \, \alpha_{q} }{\sqrt{ 1 - \eta + \eta e^{-2r} + \sigma^{2} }}\right]  , \label{P11}
\end{eqnarray}
where $\mathrm{erf}\left(z\right)$ denotes the error function:
\[
\mathrm{erf}\left(z\right)  \equiv\frac{2}{\sqrt{\pi}}\int_{0}^{z} \exp\left\{  -x^{2}\right\}  \ \mbox{d}x.
\]

This situation is identical to the one treated in Ref.~\cite{WK09}, and
the forbidden interval can be determined in a simple way by looking at the
probability of successful decoding (note that others have also considered the 
probability-of-success, or error, criterion~\cite{RAC06,PK09a}).
The probability of success is defined as
\begin{equation}\label{Pedef}
{P}_{s}\equiv P_{X}(0)P_{Y|X}(0|0)+P_{X}(1)P_{Y|X}(1|1) \, .
\end{equation}

Setting $P_{X}(0)=\wp$ and $P_{X}(1)=(1-\wp)$, the probability of success is as
follows:
\begin{equation}
\fl {P}_{s} = \frac{1}{2} + 
\frac{1}{2}\wp\ \mathrm{erf}\left( \frac{ \theta + \sqrt{\eta} \, \alpha_{q} }{\sqrt{ 1 - \eta + \eta e^{-2r} + \sigma^{2} }}\right) 
-\frac{1}{2}(1-\wp)\ \mathrm{erf}\left(  \frac{ \theta - \sqrt{\eta} \, \alpha_{q} }{\sqrt{ 1 - \eta + \eta e^{-2r} + \sigma^{2} }}\right). \label{pe1}
\end{equation}
Our goal is to study the dependence of the success probability on the noise variance $\sigma^2/2$.
This leads us to the following proposition:

\begin{proposition}
[The forbidden interval]\label{prop:unassisted-comm} The probability of success $P_{s}$
shows a non-monotonic
behavior versus $\sigma$ iff $\theta\notin\lbrack\theta_{-},\theta_{+}]$,
where $\theta_{\pm}$ are the two roots of the following equation:
\begin{equation}\label{Eq-P1}
\frac{\wp ( \theta + \sqrt{\eta} \, \alpha_{q} )}{(1-\wp)( \theta - \sqrt{\eta} \, \alpha_{q} )} = 
\exp\left[  \frac{ 4 \sqrt{\eta} \, \alpha_{q} \theta }{1-\eta+e^{-2r}\eta }\right]  \, ,
\end{equation}
with $\theta_{-} \leqslant - \sqrt{\eta} \, \alpha_{q} < \sqrt{\eta} \, \alpha_{q} \leqslant \theta_{+}$.
\end{proposition}

\begin{proof}
We consider ${P}_{s}$ as a function of $\sigma^{2}$. 
In order to have a non-monotonic behavior for ${P}_{s}(\sigma^{2})$, 
we must check for the presence of a local maximum for positive values of $\sigma$. 
By imposing
\begin{equation}\label{derpe1}
\frac{d{P}_{s}(\sigma^{2})}{d\sigma^{2}} = 0
\end{equation}
we obtain the following expression for the critical value of $\sigma$:
\begin{equation}\label{smin1}
\sigma^{2}_\ast = - 1 + \eta - e^{-2r}\eta +
\frac{ 4 \sqrt{\eta} \, \alpha_{q} \theta }{\ln\left[  \frac{ \wp ( \theta + \sqrt{\eta} \, \alpha_{q} ) } { (1-\wp) ( \theta - \sqrt{\eta} \, \alpha_{q} ) } \right]  } \, . 
\end{equation}
The probability of success is a non-monotonic function of $\sigma^2$ iff $\sigma^{2}_\ast > 0$.
This inequality is verified for $\theta\notin\lbrack\theta_{-},\theta_{+}]$, where 
$\theta_\pm$ are the unique solutions of the equation $\sigma^{2}_\ast = 0$, i.e., equation (\ref{Eq-P1}).
Finally, we notice that equation (\ref{Eq-P1}) implies $\frac{\theta + \sqrt{\eta} \, \alpha_{q}}{\theta - \sqrt{\eta} \, \alpha_{q}} > 0$,
i.e., $\theta \not\in [ - \sqrt{\eta} \, \alpha_q , \sqrt{\eta} \, \alpha_q ]$, which implies 
$\theta_{-} \leqslant - \sqrt{\eta} \, \alpha_{q} < \sqrt{\eta} \, \alpha_{q} \leqslant \theta_{+}$.
\end{proof}

The above proposition improves upon the theorem from Ref.~\cite{WK09}\ in
several important ways, due to the assumption that the noise is Gaussian,
allowing us to analyze it more carefully. First, (\ref{smin1}) gives the
optimal value of the noise that leads to the maximum success probability if the
threshold is outside of the forbidden interval (though, note that other works
have algorithms to learn the optimal noise parameter \cite{MK98}). Second,
there is no need to consider an infinite-squeezing limit, as was the case in
Ref.~\cite{WK09}, in order to guarantee the non-monotonic signature of SR.

As an example, figure \ref{Pevssig} shows the probability of success 
$P_s$ as a function of $\sigma$ for various values of the threshold 
$\theta$ around the high signal level $\sqrt{\eta} \, \alpha_q$. 
Identical behavior can be found for values of $\theta$ around the low 
signal level~$- \sqrt{\eta} \, \alpha_q$.

\begin{figure}
\centering\includegraphics[width=0.45\textwidth] {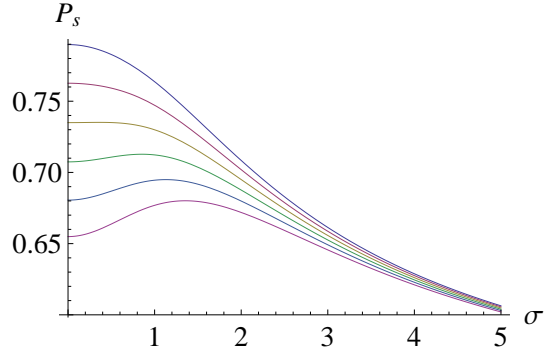} \caption{
The probability of success $P_{s}$, equation (\ref{pe1}), (corresponding to the choice $\wp=1/2$)
versus $\sigma$, for the noise-assisted threshold
detection. The values of the parameters are $\eta=0.8$, $\alpha_{q}=1$ and
$r=0$, giving $\theta_{\pm}\approx\pm 0.96$ after applying
Proposition~\ref{prop:unassisted-comm}. Curves from top to bottom correspond
respectively to $\theta=0.85$, $0.95$ (inside the forbidden interval), $1.05$, 
$1.15$, $1.25$, $1.35$ (outside the forbidden interval). 
Due to symmetry, we also have the same plot for the values $\theta=-0.85$, 
$-0.95$, $-1.05$, $-1.15$, $-1.25$, $-1.35$.} 
\label{Pevssig}
\end{figure}

Figure~\ref{thevsr} plots the forbidden interval in the $\theta,r$
plane. We can see that increasing the squeezing level reduces the
width of the forbidden interval up to $2\sqrt{\eta} \, \alpha_{q}$. 
Similarly, figure \ref{thevsal} plots the forbidden interval in the $\theta
,\alpha_{q}$ plane. Increasing the
amplitude enlarges the width of the forbidden interval up to $2\sqrt{\eta} \, \alpha_{q}$.

\begin{figure}
\centering\includegraphics[width=0.45\textwidth] {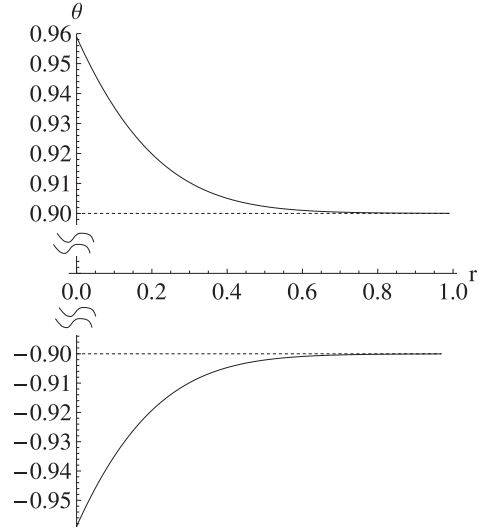}\caption{Forbidden
interval (area between upper and lower curves) in the $\theta,r$ plane, drawn according
to Proposition \ref{prop:unassisted-comm}. The top 
(resp. bottom) solid line corresponds to $\theta_{+}$ (resp. $\theta_-$) while the top 
(resp. bottom) dashed line corresponds to $\sqrt{\eta} \, \alpha_q$ (resp. $-\sqrt{\eta} \, \alpha_q$). 
The values of the other parameters are $\wp=1/2$, $\eta=0.8$, and $\alpha_{q}=1$.}
\label{thevsr}
\end{figure}

\begin{figure}
\centering\includegraphics[width=0.45\textwidth] {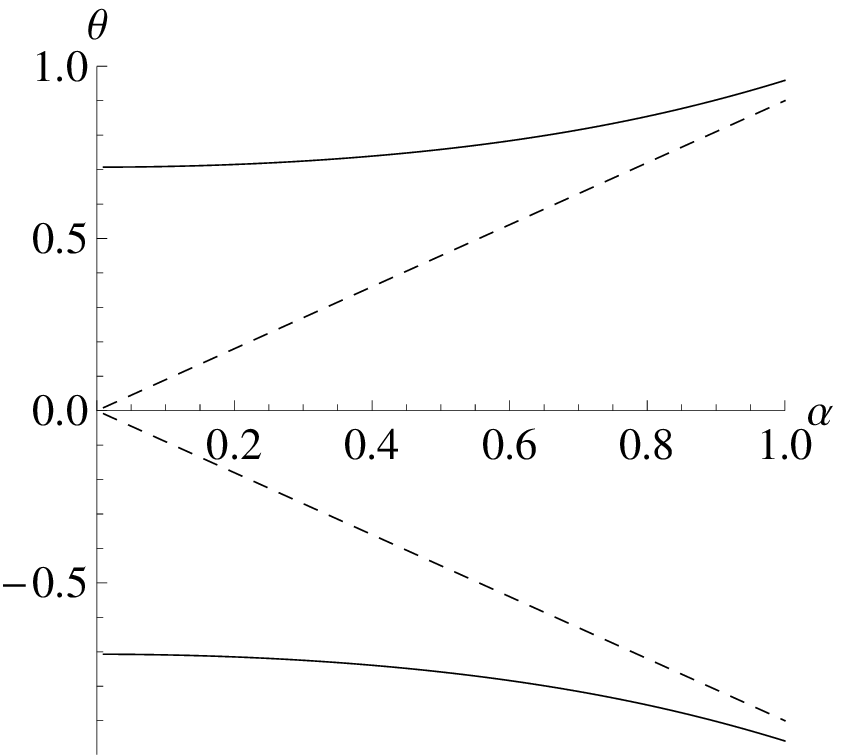}
\caption{Forbidden interval (area between upper and lower curves) in the 
$\theta,\alpha_q$ plane, drawn according to Proposition \ref{prop:unassisted-comm}. 
The top (resp. bottom) solid line corresponds to $\theta_{+}$ 
(resp. $\theta_-$), while the top (resp. bottom) dashed line corresponds to $\sqrt{\eta} \, \alpha_q$ 
(resp. $-\sqrt{\eta} \, \alpha_q$). The values of the other parameters are $\wp=1/2$, 
$\eta=0.8$, and $r=0$.}
\label{thevsal}
\end{figure}

Finally, notice that Proposition \ref{prop:unassisted-comm} holds true even if 
the noise $\nu_q$ is introduced at the sender's end.

\section{Stochastic Resonance in Entanglement-Assisted Classical Communication}

Let us consider the same channel as in the previous section, but we now assume
that the sender and receiver share an entangled state, namely, a two-mode
squeezed vacuum~\cite{G05}, before communication begins. This situation is
somehow similar to the communication scenario in super-dense
coding~\cite{BW92-BK00}, with the exception that we have continuous variable
systems and thresholding at the receiver. Let mode~1 (resp.~2) belong to the
sender (resp.~receiver). The sender displaces her share of the entanglement by
the complex number $-\alpha_{q}\left(  -1\right)  ^{X_{q}}-i\alpha_{p}\left(
-1\right)  ^{X_{p}}$ in order to transmit the two bits $X_{q}$ and $X_{p}$.
The resulting displaced squeezed vacuum operators are as follows:
\begin{eqnarray}
& &  (\hat{q}_{1}-\hat{q}_{2})e^{-r}-\alpha_{q}(-1)^{X_{q}} \, , \label{eq:qmes1}\\
& &  (\hat{p}_{1}+\hat{p}_{2})e^{-r}-\alpha_{p}(-1)^{X_{p}} \, , \label{eq:qmes2}
\end{eqnarray}
where $\hat{q}$, $\hat{p}$ are the position, momentum quadrature operators,
$r \geqslant 0$ is the squeezing strength, $\alpha_{q},\alpha_{p}\in\mathbbm{R}$ are
the displacement amplitudes, and $X_{q},X_{p}\in\{0,1\}$ are binary random variables.

Since $\hat{q}_{1}-\hat{q}_{2}$ commutes with $\hat{p}_{1}+\hat{p}_{2}$, it
suffices to analyze the output for (\ref{eq:qmes1}). After the sender
transmits her share of the entanglement through a lossy bosonic channel with
transmissivity $\eta\in\left(0,1\right)$, the operator describing their
state is as follows:
\[
\left( \sqrt{\eta} \, \hat{q_{1}} - \hat{q}_{2} \right) e^{-r} - 
\sqrt{\eta} \, \alpha_{q}(-1)^{X_{q}} + \sqrt{1-\eta} \, \hat{q}_{E} \, ,
\]
where $\hat{q}_{E}$ is the position quadrature operator of the environment
mode (assumed to be in the vacuum state for the sake of simplicity).

At the receiver's end, let us again consider the possibility of adding a
random, Gaussian-distributed displacement $\nu_{q}\in\mathbbm{R}$ to the
arriving state. Then the output observable becomes as follows:
\[
\left( \sqrt{\eta} \, \hat{q_{1}} - \hat{q}_{2} \right) e^{-r} - 
\sqrt{\eta} \, \alpha_{q}(-1)^{X_{q}} + \sqrt{1-\eta} \, \hat{q}_{E} + \nu_{q} \, .
\]

Repeating the steps of Section~\ref{sec:unassisted-comm}, we have
\[
S_{X_{q}} = N -\sqrt{\eta} \, \alpha_{q}(-1)^{X_{q}} \, ,
\]
where now
\begin{equation}\label{eq:all-noise2}
N\equiv\left( \sqrt{\eta} \, \hat{q_{1}}-\hat{q}_{2}\right) e^{-r}+\sqrt{1-\eta} \, \hat{q}_{E} + \nu_{q} \, .
\end{equation}
and
\begin{equation}\label{eq:noise-density2}
P_{N} = {\mathcal{N}}(0,\eta e^{-2r}/2) \circ {\mathcal{N}}(0,e^{-2r}/2) \circ
{\mathcal{N}}(\left(  1-\eta\right)/2) \circ {\mathcal{N}}(0,\sigma^{2}_q/2) \, ,
\end{equation}
so that
\begin{equation}\label{pX2}
P_{N} =\mathcal{N}\left( 0 , ( 1 - \eta + (1+\eta)e^{-2r} + \sigma^{2}_q )/2\right) \, .
\end{equation}

The receiver then thresholds the measurement result with a threshold
$\theta_{q}\in\mathbbm{R}$, and he retrieves a random bit $Y_{q}$ where
\begin{equation}\label{Y2}
Y_{q}\equiv H\left( N - \sqrt{\eta} \, \alpha_{q}(-1)^{X_{q}} - \theta_{q} \right) \, ,
\end{equation}
and $H$ is the unit Heaviside step function.

Proceeding as in Section~\ref{sec:unassisted-comm}, we obtain the following
input/output conditional probabilities
\begin{eqnarray}
P_{Y_{q}|X_{q}}(0|0) & = & \frac{1}{2}+\frac{1}{2}\mathrm{erf}\left[
\frac{\theta_{q} + \sqrt{\eta} \, \alpha_{q}}{\sqrt{ 1 - \eta + (1+\eta)e^{-2r} + \sigma^{2}_q}}\right]  \, ,\\
P_{Y_{q}|X_{q}}(1|1) & = & \frac{1}{2}-\frac{1}{2}\mathrm{erf}\left[
\frac{\theta_{q} - \sqrt{\eta} \, \alpha_{q}}{\sqrt{ 1 - \eta + (1+\eta)e^{-2r} + \sigma^{2}_q }}\right]  \, .
\end{eqnarray}

Then, writing $P_{X_{q}}(0)=\wp_{q}$ and $P_{X_{q}}(1)=1-\wp_{q}$, the
probability of success reads
\begin{eqnarray}
{P}_{s,q} & = & \frac{1}{2}+\frac{1}{2}\wp_{q} \,
\mathrm{erf}\left(\frac{\theta_{q}+\sqrt{\eta} \, \alpha_{q}}{\sqrt{ 1 - \eta + (1+\eta )e^{-2r} + \sigma^{2}_q }}\right) \nonumber \\
& & -\frac{1}{2}(1-\wp_{q}) \, \mathrm{erf}\left( \frac{\theta_{q}-\sqrt{\eta} \, \alpha_{q}}{\sqrt{ 1 - \eta + (1+\eta)e^{-2r} + \sigma^{2}_q }}\right)  \, . \label{pe2}
\end{eqnarray}
Analogously, for the quadrature $\hat{p}_{1}+\hat{p}_{2}$, we have
\begin{eqnarray}
{P}_{s,p} & = & \frac{1}{2}+\frac{1}{2}\wp_{p}\,\mathrm{erf}\left(
\frac{\theta_{p}+\sqrt{\eta} \, \alpha_{p}}{\sqrt{ 1 - \eta + (1+\eta)e^{-2r} + \sigma^{2}_p }}\right) \nonumber \\
& & -\frac{1}{2}(1-\wp_{p})\,\mathrm{erf}\left( \frac{\theta_{p}-\sqrt{\eta} \, \alpha_{p}}{\sqrt{ 1 - \eta + (1+\eta)e^{-2r} + \sigma^{2}_p }}\right)  \, . \label{pe3}
\end{eqnarray}
We have assumed that the noise terms added to the quadratures 
$\hat{q}_{1}-\hat{q}_{2}$ and $\hat{p}_{1}+\hat{p}_{2}$ have variances $\sigma_q^2/2$ 
and $\sigma_p^2/2$ respectively.

We finally arrive at the following proposition:

\begin{proposition}[The forbidden rectangle] \label{forbrect}
The probability of success $P_s=P_{s,q} P_{s,p}$
shows a non-monotonic behavior 
vs $\sigma_q$, $\sigma_p$ iff $\theta_q\notin[\theta_{q-},\theta_{q+}]$ or 
$\theta_p\notin[\theta_{p-},\theta_{p+}]$ where $\theta_{\bullet\pm}$ are the roots 
of the following equation:
\[
\frac{\wp_{\bullet}( \theta_{\bullet} + \sqrt{\eta} \, \alpha_{\bullet} )}{(1-\wp_{\bullet})( \theta_{\bullet} - \sqrt{\eta} \, \alpha_{\bullet} )}=
\exp\left[\frac{4\sqrt{\eta} \, \alpha_{\bullet}\theta_{\bullet}}{1-\eta+(1+\eta)e^{-2r}}\right],
\]
with 
$\theta_{\bullet -}\le - \sqrt{\eta} \, \alpha_{\bullet} < + \sqrt{\eta} \, \alpha_{\bullet} \le \theta_{\bullet +}$
(here $\bullet$ stands for either $q$ or $p$).
\end{proposition}

\begin{proof}
The proof can be obtained from that of Proposition \ref{prop:unassisted-comm}
after replacing $-\eta + \eta e^{-2r}$ with $-\eta + (1+\eta) e^{-2r}$.
\end{proof}

It is worth remarking in the above proposition that either of the conditions
$\theta_q\notin[\theta_{q-},\theta_{q+}]$ or $\theta_p\notin[\theta_{p-},\theta_{p+}]$ 
 (or both) have to be satisfied in order to have a non monotonic behavior for 
$P_s$.
This follows because $P_s$ is a function of two variables $\sigma_q^2$ 
and $\sigma_p^2$, hence it suffices that the partial derivative with respect to one of them 
has a maximum to have a non-monotonic behavior. 
As consequence we have a \textquotedblleft forbidden rectangle\textquotedblright\ 
rather than \textquotedblleft forbidden stripes\textquotedblright\ in the $\theta_q,\theta_p$ 
plane. 

It is also worth noticing the difference of the equation in the above proposition with 
that in the proposition of the previous section. Here the squeezing factor $e^{-2r}$ 
is multiplied by $(1+\eta)$ rather than $\eta$. This is because we now have two 
squeezed modes, one of which is attenuated by the lossy channel.

Finally, notice that Proposition \ref{forbrect} holds true even if the noise $\nu_q$ 
is introduced at the sender's end.

\section{Stochastic Resonance in Channel Discrimination}

Let us now consider two lossy quantum channels with transmissivities $\eta
_{0},\eta_{1}\in(0,1)$ (suppose, without loss of generality, $\eta_{0}>\eta_{1}$). Our aim is to
distinguish them. Differently from previous works, here we do not 
optimize over all possible decoding strategies \cite{lossD}, but concentrate 
on a given threshold-detection scheme.
Then, suppose to use as probe (input) a squeezed and
displaced vacuum operator
\begin{equation}\label{eq:qmessage3}
\hat{q}e^{-r}+\alpha_{q} \,,
\end{equation}
where $\hat{q}$ is the position quadrature operator, $r$ is the squeezing
parameter, and $\alpha_{q}\in\mathbbm{R}$ the displacement amplitude.

The transmission through the lossy channel with transmissivity $\eta_{X}$,
$X=0,1$, can be considered as the encoding of a binary random variable
$X=0,1$ occurring with probability $P_{X}$. The output observable after
transmission is then
\[
\sqrt{\eta_{X}} \left( \alpha_{q} + \hat{q}e^{-r} \right) + \sqrt{1-\eta_{X}} \, \hat{q}_{E} \, ,
\]
where $\hat{q}_{E}$ is the position quadrature operator of the environment
mode (assumed to be in the vacuum state for the sake of simplicity).

At the receiver's end, we consider the addition of noise, modeled by a random,
Gaussian-distributed displacement $\nu_{q}\in\mathbbm{R}$.
Then the output observable becomes
\[
\sqrt{\eta_{X}} \left( \alpha_{q}+\hat{q}e^{-r}\right) + \sqrt{1-\eta_{X}} \, \hat{q}_{E} + \nu_{q} \, .
\]

Upon measurement of the position quadrature operator, the signal value is
\begin{equation}\label{signal3}
S_{X} = \sqrt{\eta_{X}} \left( \alpha_{q}+qe^{-r}\right) + \sqrt{1-\eta_{X}} \, q_{E}+\nu_{q} \, .
\end{equation}
We define a conditional random variable $N|X$ summing all noise terms:
\begin{equation}\label{eq:all-noise3}
N|X \equiv \sqrt{\eta_{X}} \, q e^{-r} + \sqrt{1-\eta_{X}} \, q_{E} + \nu_{q} \, .
\end{equation}
The density $P_{N|X}\left(  n|x\right)$ of the random variable $N|X$ is
\begin{equation}\label{eq:noise-density3}
P_{N|X} = {\mathcal{N}}(0,\eta_{x}e^{-2r}/2) \circ {\mathcal{N}}(0,\left(  1-\eta_{x}\right)/2)
\circ {\mathcal{N}}(0,\sigma^{2}/2) \, .
\end{equation}

Notice that the noise term (\ref{eq:all-noise3}) explicitly depends on the
encoded value $X$ and so does its probability density. From
(\ref{eq:noise-density3}), we explicitly obtain
\begin{equation}\label{pX3}
P_{N|X} = \mathcal{N}\left(  0,( 1 - \eta_{x} + \eta_{x}e^{-2r} + \sigma^{2} )/2 \right) \, .
\end{equation}

The output signal (\ref{signal3}) can now be written as
\[
S_{X} = \sqrt{\eta_{X}} \, \alpha_{q}+N|X \, .
\]
The receiver then thresholds the measurement result with a threshold
$\theta\in\mathbbm{R}$ to retrieve a random bit $Y$ where
\[
Y \equiv H\left(  \theta - \sqrt{\eta_{X}} \, \alpha_{q} - N|X \right) \, ,
\]
and $H$ is the unit Heaviside step function. In this case, the receiver
assigns $Y=1$ if the output signal $S_{X}$ is smaller that the threshold, 
and assigns $Y=0$ otherwise.

The final detected bit $Y$ should be the same as the encoded bit $X$. 
Hence, the probability of success reads like (\ref{Pedef}) where now
\begin{eqnarray}
P_{Y|X}(0|0) & = & \int_{-\infty}^{+\infty} \hspace{-0.4cm} 
\left[ 1 - H\left( \theta -\sqrt{\eta_{0}} \, \alpha_{q}-n\right) \right] P_{N|X} \left( n|0\right) dn \, , \nonumber \\
P_{Y|X}(1|1) & = & \int_{-\infty}^{+\infty} \hspace{-0.4cm} 
H\left( \theta - \sqrt{\eta_{1}} \, \alpha_{q} - n \right) P_{N|X} \left( n|1 \right) dn \, . \nonumber
\end{eqnarray}

Using (\ref{pX3}) we obtain
\begin{eqnarray}
P_{Y|X}(0|0) & = & \frac{1}{2}\left[  1-\mathrm{erf}\left(  \frac{\theta - 
\sqrt{\eta_{0}} \, \alpha_{q}}{\sqrt{ 1 - \eta_{0} + \eta_{0}e^{-2r} + \sigma^{2} }}\right)
\right]  \, , \\
P_{Y|X}(1|1) & = & \frac{1}{2}\left[  1+\mathrm{erf}\left(  \frac{\theta - 
\sqrt{\eta_{1}} \, \alpha_{q}}{\sqrt{ 1 - \eta_{1} + \eta_{1}e^{-2r} + \sigma^{2} }}\right)
\right]  \, .
\end{eqnarray}

Then, writing $P_{X}\left(  0\right)  =\wp$ and $P_{X}\left(  1\right) = 1-\wp$, we get
\begin{eqnarray}
{P}_{s} & = & \frac{1}{2}-\frac{1}{2}\wp\,\mathrm{erf}\left( 
\frac{ \theta - \sqrt{\eta_{0}} \, \alpha_{q} }{\sqrt{ 1 - \eta_{0} + \eta_{0}e^{-2r} + \sigma^{2} }}\right) \nonumber \\
& & +\frac{1}{2}(1-\wp)\,\mathrm{erf}\left(  \frac{ \theta - \sqrt{\eta_{1}} \, \alpha_{q} }{\sqrt{ 1 - \eta_{1} + \eta_{1}e^{-2r} + \sigma^{2} }}\right) \, . \label{pe4}
\end{eqnarray}
Our aim is to analyze the probability of success as a function of the noise
variance, for given values of the parameters $\alpha_{q}$, $r$, $\theta$, 
$\eta_{0}$, $\eta_{1}$.

In the simplest case of $r=0$, we have the following proposition:

\begin{proposition}[The forbidden interval]\label{forbdiscr}
The probability of success $P_{s}$ shows a non-monotonic
behavior as a function of $\sigma$ iff $\theta\notin\lbrack\theta_{-},\theta_{+}]$, 
where $\theta_{\pm}$ are the two roots of the following
equation:
\begin{equation}\label{cond_cd}
\frac{\wp( \theta - \sqrt{\eta_{0}} \, \alpha_{q})}{(1-\wp)( \theta - \sqrt{\eta_{1}} \, \alpha_{q} )}=
\exp\left[  \alpha_{q}^{2}(\eta_{0}-\eta_{1})-2\alpha_{q}\theta(\sqrt{\eta_{0}}-\sqrt{\eta_{1}})\right]  \, ,
\end{equation}
such that $\theta_{-} \leqslant \sqrt{\eta_{1}} \, \alpha_{q} < \sqrt{\eta_{0}} \, \alpha_{q} \leqslant \theta_{+}$.
\end{proposition}

\begin{proof}
We consider the probability of success as a function of $\sigma^{2}$. 
In order to have a non-monotonic behavior for ${P}_{s}(\sigma^{2})$, 
we must check for the presence of a local maximum. By solving
\begin{equation}\label{derpe3}
\frac{d{P}_{s}(\sigma^{2})}{d\sigma^{2}}=0 \, ,
\end{equation}
we obtain the following expression for the critical value of $\sigma^2$:
\begin{equation}\label{smin3}
\sigma_{\ast}^{2}=-1+\frac{\alpha_{q}^{2}(\eta_{0}-\eta_{1})-2\alpha_{q}
\theta(\sqrt{\eta_{0}}-\sqrt{\eta_{1}})}{\ln\left[  \frac{ \wp(\theta - \sqrt{\eta_{0}} \,
\alpha_{q})}{(1-\wp)( \theta - \sqrt{\eta_{1}} \, \alpha_{q} )} \right]} \,.
\end{equation}
The condition for non-monotonicity of the probability of success as a function
of $\sigma^2$, $\sigma_{\ast}^{2} > 0$, is verified iff $\theta \not\in [\theta_-, \theta_+]$,
where $\theta_\pm$ are the roots of $\sigma_{\ast}^{2}=0$, i.e., equation (\ref{cond_cd}).
Finally, equation (\ref{cond_cd}) implies $\frac{ \theta - \sqrt{\eta_{0}} \,
\alpha_{q}}{\theta - \sqrt{\eta_{1}} \, \alpha_{q}} > 0$, which in turn yields 
$\theta_{-} \leqslant \sqrt{\eta_{1}} \, \alpha_{q} < \sqrt{\eta_{0}} \, \alpha_{q} \leqslant \theta_{+}$.
\end{proof}

If $r\neq0$, (\ref{derpe3}) is not algebraic---hence we did not succeed in 
providing an analytical expression for $\sigma_{\ast}^{2}$. However, numerical
investigations show a qualitative behavior of the forbidden interval's
boundaries identical to that shown in figures \ref{thevsr} and \ref{thevsal}
(notice that here $-\sqrt{\eta} \, \alpha_{q}$ and $\sqrt{\eta} \, \alpha_{q}$ are
replaced by $\sqrt{\eta_{1}} \, \alpha_{q}$ and $\sqrt{\eta_{0}} \, \alpha_{q}$, respectively).

Finally, notice that Proposition \ref{forbdiscr} holds true even if the 
noise $\nu_q$ is introduced at the sender's end.

\section{Stochastic Resonance in Quantum Communication}\label{Sec:quantum}

Let us now consider a setting in which the SR effect can occur in
the transmission of a qubit ($Q$).
The aim is to first encode a qubit state into a bosonic mode ($B$) state, 
send it through the lossy channel, and finally coherently decode, with a threshold mechanism, 
the output bosonic mode state into a qubit system at the receiving end. We should qualify that it is unclear
to us whether one would actually exploit the encodings and decodings given in this section, but
regardless, the setting given here provides a novel scenario in which the SR effect can occur
for a quantum system. We might consider the development in this section to be a coherent version
of the settings in the previous sections. Also, it is in the spirit of a true ``quantum stochastic resonance''
effect hinted at in Ref.~\cite{B05}.

We work in the Schr\"{o}dinger picture, and consider an initial state 
$|\varphi\rangle_Q\otimes |0\rangle_B$, where
$|\varphi\rangle_Q=a|0\rangle_Q+b|1\rangle_Q$ is an arbitrary qubit state 
and $|0\rangle_B$ is the zero-eigenstate of the position-quadrature operator of the bosonic field. 
Here for the sake of simplicity we are going to work with infinite-energy position 
eigenstates rather than with squeezed-coherent states.
 
Suppose that the encoding takes place through the following unitary controlled-operations:
\begin{eqnarray}
U_1^{QB} & = & \left\vert 0\right\rangle_Q \left\langle 0\right\vert\otimes
e^{-i\hat{p}_B x_0}+\left\vert 1\right\rangle_Q \left\langle 1\right\vert \otimes
e^{i\hat{p}_B x_0} \, ,\label{U1} \\
U_2^{QB} & = & I_{Q} \otimes \int_{x \geqslant 0}\left\vert x\right\rangle_B \left\langle
x\right\vert dx + X_{Q}\otimes\int_{x<0}\left\vert x\right\rangle_B
\left\langle x\right\vert dx \, , \label{U2}
\end{eqnarray}
where $I_{Q} = |0\rangle_Q\langle 0| + |1\rangle_Q\langle 1|$ and $X_{Q} = |0\rangle_Q\langle 1| + |1\rangle_Q\langle 0|$,
$\hat{p}_B$ denotes the canonical momentum operator of the bosonic system,
$|\pm x_0\rangle$ are the generalized eigenstates of the canonical position operator
$\hat{q}_B$, and we assume, without loss of generality, $x_{0}\in\mathbbm{R}_+$. This encoding is a coherent version
of encoding a binary number into an analog signal.

The effect of such operations on the initial states is 
\begin{eqnarray}
U_2^{QB}U_1^{QB}|\varphi\rangle_Q\otimes |0\rangle_B
= |0\rangle_Q\otimes \left(a |x_0\rangle_B
+b |-x_0\rangle_B\right).
\end{eqnarray}
Now, with the bosonic mode state factored out from the qubit state, it can be sent through 
the lossy channel. For the sake of analytical investigation, we consider a channel with unit trasmittivity
(an identity channel). 
Then, the output state simply reads 
\[
\rho_B= \left(a |x_0\rangle
+b |-x_0\rangle\right)_B\left(\overline{a}\, \langle x_0|
+\overline{b}\, \langle -x_0|\right),
\]
At this point we consider the possibility of adding Gaussian noise before the (threshold) decoding stage. 
This is modeled as a Gaussian-modulated displacement of the quadrature $\hat{q}_B$. 
The resulting state will be
\begin{equation}\label{rhoBp}
\rho_{B^\prime}=\int dq \, \, {\mathcal{N}}(0,\sigma^2) D(q,0) \rho_B D^{\dag}(q,0) \, ,
\end{equation}
where $D(q,0)$ is the displacement operator (displacing only in the $\hat{q}_B$ direction), 
and ${\mathcal{N}}(0,\sigma^2)$ is a zero-mean, Gaussian distribution with variance $\sigma^2$.

Now, the state $\rho_{B^\prime}$ is decoded into a qubit system 
$Q^{\prime}$ initially prepared in the state $|0\rangle_{Q^{\prime}}$ through 
the following controlled-unitary operations involving a coherent threshold mechanism
\begin{eqnarray}
V_1^{Q^{\prime}B^{\prime}} & = & I_{Q}\otimes\int_{x \geqslant \theta}\left\vert x\right\rangle_B \left\langle
x\right\vert dx 
+X_{Q}\otimes\int_{x<\theta}\left\vert x\right\rangle_B
\left\langle x\right\vert dx \, , \label{V1} \\
V_2^{Q^{\prime}B^{\prime}} & = & \left\vert 0\right\rangle_Q \left\langle 0\right\vert\otimes
e^{i\hat{p}_B x_0}+\left\vert 1\right\rangle_Q \left\langle 1\right\vert \otimes
e^{-i\hat{p}_B x_0} \, . \label{V2}
\end{eqnarray}
Clearly, if $\theta=0$ the decoding unitaries (\ref{V1}), (\ref{V2}) are the inverse 
of the encoding ones (\ref{U1}), (\ref{U2}). They hence allow unit fidelity encoding/decoding 
if $\sigma^2=0$. However, if $\theta\neq 0$, there could be a nonzero optimal value of $\sigma^2$.

The final qubit state is
\begin{eqnarray}
\rho_{Q^{\prime}} & = & \mathcal{E}(|\varphi\rangle_Q\langle\varphi|) \label{sqp} \\
& = & {\rm Tr}_{B^{\prime}} \left\{ V_2^{Q^{\prime}B^{\prime}} V_1^{Q^{\prime}B^{\prime}} 
 \left(|0\rangle_Q^{\prime}\langle 0|\otimes \rho_{B^{\prime}}\right) \left( V_1^{Q^{\prime}B^{\prime}}\right)^{\dag}
 \left(V_2^{Q^{\prime}B^{\prime}}\right)^{\dag} \right\} \\
& = & \left[ |a|^2 (1-\Pi_<) + |b|^2 \Pi_>) \right] |0\rangle_Q\langle 0| 
+ \left[ |a|^2 \Pi_< + |b|^2 (1-\Pi_>) \right] |1\rangle_Q\langle 1| \nonumber\\
& & +(1-\Pi_<-\Pi_>) \left( a \bar{b}|0\rangle_Q\langle 1| + \bar{a}b |1\rangle_Q\langle 0| \right) \, , \label{sigmaout}
\end{eqnarray}
where 
\begin{eqnarray}
\Pi_< & = & \frac{1}{2} + \frac{1}{2} {\rm erf}\left( \frac{\theta-x_0}{\sqrt{2\sigma^2}} \right) \, , \\
\Pi_> & = & \frac{1}{2} - \frac{1}{2} {\rm erf}\left( \frac{\theta+x_0}{\sqrt{2\sigma^2}} \right) \, .
\end{eqnarray}
Then we can calculate the average channel fidelity
\begin{equation}\label{Pefid}
\langle \mathcal{F} \rangle = \int d\varphi \,\, {}_Q\langle\varphi|\rho_{Q^{\prime}}|\varphi\rangle_Q \, ,
\end{equation}
where $d\varphi$ is the uniform measure induced by the Haar measure on $\mathrm{SU}(2)$.

Using (\ref{rhoBp}), (\ref{V1}), (\ref{V2}), (\ref{sqp}) and (\ref{Pefid}) we finally obtain
\begin{equation}\label{fidqc}
\langle \mathcal{F} \rangle  = \frac{1}{2} + \frac{1}{4}
\left[ {\rm erf}\left(\frac{\theta+x_0}{\sqrt{2\sigma^2}}\right) - {\rm erf}\left(\frac{\theta-x_0}{\sqrt{2\sigma^2}}\right) \right].
\end{equation}

Then, we have the following proposition:

\begin{proposition}\label{forbquantum}
[The forbidden interval] The average channel fidelity $\langle \mathcal{F} \rangle$ shows a
non-monotonic behavior as a function of $\sigma$ iff $\theta \notin [- x_0 , x_0 ]$.
\end{proposition}

\begin{proof}
We consider $\langle \mathcal{F} \rangle$ as a function of $\sigma^2$. 
In order to have a non-monotonic behavior, we must check for the presence of a local maximum
for $\sigma^2 > 0$.
The condition $\frac{d\langle F \rangle}{d\sigma}=0$ yields the 
following expression for the critical value of $\sigma^2$
\begin{equation}
\sigma^2_\ast = \frac{2\theta x_0}{\ln\left(\frac{\theta+x_0}{\theta-x_0}\right)} \, .
\end{equation}
We hence conclude that the average fidelity is a non-monotonic function of 
$\sigma$ iff $\theta \notin [ - x_0 , x_0 ]$.
\end{proof}

\begin{figure}
\centering\includegraphics[width=0.45\textwidth] {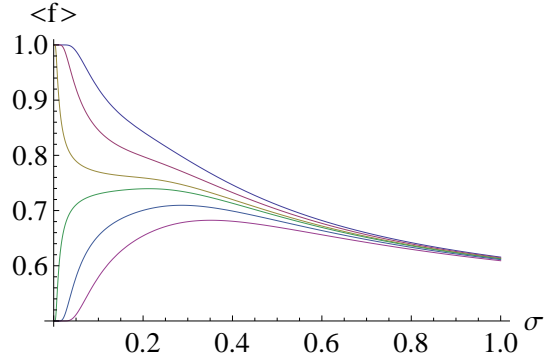}
\caption{The average channel fidelity $\langle \mathcal{F} \rangle$, equation (\ref{fidqc}) 
vs $\sigma$ for $x_0=0.3$ and several values of $\theta$, inside and outside the forbidden interval. 
From top to bottom, $\theta=0.20$, $\theta=0.25$, $\theta=0.29$
(inside the forbidden interval), $\theta=0.31$, $\theta=0.35$, $\theta=0.40$
(outside the forbidden interval).}
\label{fidelity}
\end{figure}

Figure~\ref{fidelity} shows $\langle \mathcal{F} \rangle$ as a function of $\sigma$ for 
a given value of $x_0$ and several values of $\theta$, both inside and outside
the forbidden interval. 
A non-monotonic behavior is observed in the latter cases. 
It is worth noticing that the presence of noise can augment the average 
channel fidelity above the value of $2/3$, which is the maximum value achievable by 
measure-and-prepare protocols.
Hence, in this sense, the presence of noise can lead to a transition
from a classical to a quantum regime in the average communication fidelity.

As shown in figure \ref{quantum}, an analogous SR-like effect is observed in the same
parametric region for the logarithmic negativity \cite{P05}
\begin{equation}\label{LogNeg} 
LN=\log_2\left\{\mathrm{Tr}\left[\mathcal{E}\otimes\mathcal{I}\left(|\Psi\rangle\langle\Psi|\right)\right]^{\Gamma}\right\} \, ,
\end{equation}
where $\Gamma$ indicates the partial transpose operation, 
$\mathcal{I}$ the identity map and $|\Psi\rangle$ a maximally entangled
two-qubit Bell state. This quantity is the logarithmic negativity of the Choi-Jamiolkowski
state associated to the quantum channel \cite{J72} 
and gives an upper bound on its two-way distillable entanglement (this latter quantity
in turn equals the quantum capacity 
of a channel assisted by unbounded two-way classical communication \cite{BDSW96}).

\begin{figure}
\centering\includegraphics[width=0.45\textwidth] {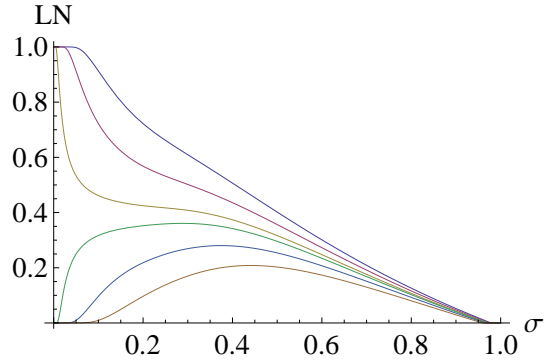}
\caption{The logarithmic negativity, equation (\ref{LogNeg}), vs $\sigma$ for $x_0=0.3$ and
several values of $\theta$ inside and outside the forbidden interval. 
From top to bottom, $\theta=0.20$, $\theta=0.25$, $\theta=0.29$
(inside the forbidden interval), $\theta=0.31$, $\theta=0.35$, $\theta=0.40$
(outside the forbidden interval).}
\label{quantum}
\end{figure}

Finally, notice that Proposition \ref{forbquantum} holds true 
even if the noise $\nu_q$ is introduced at the sender's end.

\section{Stochastic Resonance in Quantum Key Distribution}

In this section we investigate stochastic resonance effects in
quantum key distribution (see also \cite{RGK05,QKD}). An achievable rate for private classical
communication over a quantum channel is given by the following formula \cite{Dev05,CWY04}:
\begin{equation}\label{CP}
C_P = I(A:B) - I(A:E) \, ,
\end{equation}
where $I(A:B)$ and $I(A:E)$ are the maximal mutual information between
sender ($A$) and receiver ($B$) and sender and eavesdropper ($E$), respectively. In the above formula,
$A$ is a classical system, $B$ and $E$ are quantum systems, and the optimization is over all ensembles that Alice
can prepare at the input.

We consider the same encoding of a binary variable into a single bosonic mode expressed by Eq.\ (\ref{eq:qmessage1}).
Also, we consider the case in which information is transmitted through a lossy bosonic channel characterized by the
transmissivity parameter $\eta$ \cite{HW01}. The input variable after the transmission through the noisy channel is 
hence expressed by Eq.\ (\ref{lossych}).

First, suppose that the noise is added at the receiver's end.
In this case, the quantity $I(A:E)$ is not affected at all by the noise, and so 
the behavior of $C_P$ versus $\sigma$ is simply determined by $I(A:B)$. 
Since this latter relies on the probability of success given by equation (\ref{pe1}),
we are in the same situation as in Proposition \ref{prop:unassisted-comm}.
In particular, the private communication rate in Eq.\ (\ref{CP}) exhibits a 
non-monotonic behavior as a function of the noise variance if and only if
the threshold value $\theta$ lies outside of the forbidden interval 
$\lbrack\theta_{-},\theta_{+}]$, where $\theta_{\pm}$ are the two roots of equation (\ref{Eq-P1}).

Second, suppose that the noise is added at the sender's end. In this case, equation (\ref{noiseB}) changes as follows:
\begin{equation}
\sqrt{\eta}\left( \hat{q}e^{-r} - \alpha_{q}(-1)^{X} + \nu_{q}\right)  + \sqrt{1-\eta} \, \hat{q}_{E} \, .
\end{equation}
Using a threshold decoding, we get the same
expression as in equation (\ref{pe1}) for the success probability,
upon replacing $\sigma^2 \to \eta\sigma^2$. From this,
it is straightforward to calculate the mutual information $I(A:B)$.
In turn, we assume that the eavesdropper has access to the conjugate mode at the output of 
the beam-splitter transformation, and so its variable is given by
\begin{equation}\label{evemode}
\sqrt{1-\eta}\left( \hat{q}e^{-r} - \alpha_{q}(-1)^{X} + \nu_{q} \right) + \sqrt{\eta} \, \hat{q}_{E} \, .
\end{equation}
The maximum mutual information between the sender and eavesdropper is given in terms of the 
Holevo information \cite{Holevo}.
Since the average state corresponding to the variable (\ref{evemode}) 
is non-Gaussian, the analytical evaluation of its Holevo information appears not to be possible.
However, the monotonicity property of the Holevo information under composition of quantum channels
ensures that it has to be a monotonically decreasing function of the noise variance $\sigma^2/2$.
As a consequence, we expect that its contribution to the private communication rate
will increase with increasing value of the noise.
Indeed, a numerical analysis suggests that the private communication rate can exhibit a non-monotonic behavior 
as function of $\sigma$ for all values of $\theta$. Examples of this behavior are shown in figure \ref{Private}.

\medskip

We are then led to formulate the following conjecture:

\begin{conjecture}\label{forbprivate}
[The forbidden interval] The private communication rate $C_P$ shows a
non-monotonic behavior as a function of $\sigma$ 
for all $\theta\in\mathbbm{R}$ if $\nu_q$ is added at the sender's end.
\end{conjecture}

\begin{figure}
\centering
\includegraphics[width=0.45\textwidth] {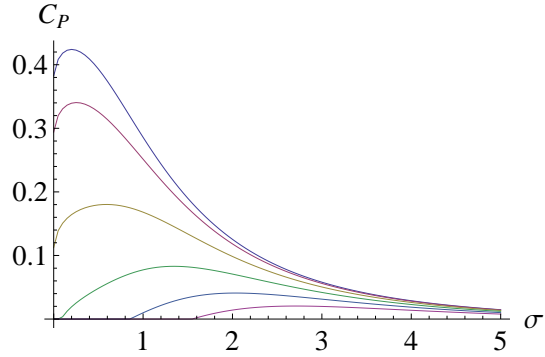}
\caption{The private communication rate (corresponding to the choice $\wp=1/2$), equation (\ref{CP}),
versus $\sigma$, for the case of noise added by the sender. 
The values of the parameters are $\eta=0.8$, $\alpha_{q}=1$ and $r=0$. 
Curves from top to bottom correspond respectively to $\theta=0$, $0.5$, $1$, $1.5$, $2$, $2.5$.} 
\label{Private}
\end{figure}

\section{Conclusion}

In conclusion, we have determined necessary and sufficient conditions for observing SR
when transmitting classical, private, and quantum information over a lossy bosonic 
channel or when discriminating lossy channels. 
Nonlinear coding and decoding by threshold mechanisms have been exploited together with the addition of Gaussian noise.

Specifically, we have considered a bit encoded into coherent
states with different amplitudes that are subsequently sent through a lossy
bosonic channel and decoded at the output by threshold measurement of their
amplitudes (without and with the assistance of entanglement shared by sender
and receiver).
We have also considered discrimination of lossy bosonic
channels with different loss parameters. In all these cases,
the performance is evaluated in terms of success probability. Since
the mutual information is a monotonic function of this probability, 
the same conclusions can be drawn in terms of mutual information.

SR effects appear whenever the threshold lies outside of the different 
forbidden intervals that we have established.
If it lies inside of a forbidden interval, then the SR effect does 
not occur. Actually, absolute maxima of success probability are obtained when the threshold is set in the 
middle of the forbidden interval.

Generally speaking, SR effects are known to improve analog-to-digital 
conversion performance \cite{Gamma}. 
In fact, if two distinct signals by continuous-to-binary conversion fall within the same interval 
they can no longer be distinguished. 
In such a situation the addition of a moderate amount of noise turns out to be useful as long as it
shifts the signals apart to help in distinguishing them.
While it is important to confirm this possibility also in the quantum framework, we have also 
shown that the same kind of effects may arise in a purely quantum framework.
Indeed, we have also considered the transmission of quantum information, represented 
by a qubit which is encoded into the state of a bosonic mode and then decoded 
according to a threshold mechanism. 
The found nonmonotonicity of the average channel fidelity and of the output 
entanglement (quantified by the logarithmic negativity) outside the forbidden interval, 
represents a clear signature of a purely quantum SR effect.

In all the above mentioned cases it does not matter whether the sender or the receiver adds the noise.
The exception occurs when the goal is to transmit private information. 
In fact, by considering achievable rates for private transmission over the lossy channel,
we have pointed out that the forbidden interval can change drastically, 
depending on whether the receiver or the sender adds noise. 
In the former case, it is exactly the same as the case of sending classical (non private) information.
In the latter case, we conjecture that it vanishes, i.e., the noise addition turns out to be beneficial always.
This feature of the private communication rate can be interpreted as a consequence of the asymmetry
between the legitimate receiver of the private information and the eavesdropper. In fact, while the
legitimate receiver is restricted to threshold detection, we have allowed the eavesdropper to use
more general detection schemes.

\end{document}